\documentclass[journal,twoside,web]{ieeecolor}
\usepackage{lcsys}
\usepackage{cite}
\usepackage{amsmath,amssymb,amsfonts}
\usepackage{graphicx}
\usepackage{algorithm,algpseudocode}
\usepackage{hyperref}
\usepackage{textcomp}

\title{Robust Output-Feedback MPC for Nonlinear Systems
  with Applications to Robotic Exploration}

\author{Scott~Brown,~\IEEEmembership{Student Member,~IEEE},
  Mohammad~Khajenejad,~\IEEEmembership{Member,~IEEE},
  Aamodh~Suresh,~\IEEEmembership{Member,~IEEE}, and
  Sonia~Mart{\'\i}nez,~\IEEEmembership{Fellow,~IEEE} \thanks{S. Brown
    and S. Mart{\'\i}nez are with the Mechanical and Aerospace
    Engineering Department of the Jacobs School of Engineering,
    University of California, San Diego, La Jolla, San Diego, CA, USA
    (e-mail: \{sab007, soniamd\}@ucsd.edu). M. Khajenejad is with the
    departments of Mechanical Engineering and Electrical and Computer
    Engineering, The University of Tulsa (e-mail:
    mohammad-khajenejad@utulsa.edu).  A. Suresh is with the DEVCOM US
    Army Research Lab, Adelphi, MD, USA (e-mail: aamodh@gmail.com).}
  \thanks{This work was partially supported by AI2C ARL grant
    W911NF-23-2-0009 and NSF grant 1947050.}  }

\newcommand{\bulletsym}{\hbox{$\bullet$}}
\newcommand{\bulletend}{\relax\ifmmode\else\unskip\hfill\fi\bulletsym}

\newcommand{\ol}[1]{\overline{#1}}
\newcommand{\ul}[1]{\underline{#1}}

\newcommand{\real}{\mathbb{R}}

\newtheorem{prop}{Proposition}

\newtheorem{thm}{Theorem}
\newtheorem{assumption}{Assumption}

\newtheorem{defn}{Definition}
\newtheorem{rem}{Remark}

\begin{document}
\maketitle
\thispagestyle{empty}

\begin{abstract}
  This paper introduces a novel method for robust output-feedback
  model predictive control (MPC) for a class of nonlinear
  discrete-time systems. We propose a novel interval-valued predictor
  which, given an initial estimate of the state, produces intervals
  which are guaranteed to contain the future trajectory of the
  system. By parameterizing the control input with an initial
  stabilizing feedback term, we are able to reduce the width of the
  predicted state intervals compared to existing methods. We
  demonstrate this through a numerical comparison where we show that
  our controller performs better in the presence of large amounts of
  noise.  Finally, we present a simulation study of a robot navigation
  scenario, where we incorporate a time-varying entropy term into the
  cost function in order to autonomously explore an uncertain area.
\end{abstract}


\newcommand{\pmbox}[1]%
{{}^\ulcorner_\llcorner\hspace{-0.2em}{#1}\hspace{-0.1em}^\urcorner_\lrcorner}
\newcommand{\X}{\mathcal{X}}
\newcommand{\U}{\mathcal{U}}
\newcommand{\E}{\mathcal{E}}

\section{Introduction}
\IEEEPARstart{M}{odel} predictive control (MPC) is a well-known and effective method
for controlling both linear and nonlinear systems subject to state and
input constraints. It has received considerable research attention
over the last three decades, which has led to the development of
several different approaches, including robust and tube-based methods
with various ways of representing the predicted trajectories
\cite{JR-DM-MD:17}.

In the context of robotics, many researchers have focused on the
problem of planning trajectories for robots operating in uncertain
environments which must be explored.  Sampling-based planners
\cite{SK-EF:11, ME-MS:14} are effective at planning long trajectories
with lower computational cost than methods which rely on an exhaustive
search. On the other hand, MPC-based planning methods continue to be
studied for use on robotic manipulators \cite{JN-JK-VB-FA-ST:20},
marine and road vehicles
\cite{HW-YS:23,SD-UM-MD-DO-TM-AM-SF:20,QS-JZ-AEK-ILJ:21}, and
multi-agent systems \cite{AT-LTB-AMJ:24}.

Much of the MPC literature assumes access to direct measurements of
the state. When these measurements are not available, however, the
problem of \emph{robustly} satisfying all state and input constraints
becomes much more challenging. This has been the focus of a large body
of research into output-feedback MPC. For example,
\cite{FDB-MAM-FA:18} develops a robust controller for linear systems
based on moving-horizon estimation.

A straightforward approach to output-feedback MPC is to incorporate an
observer to generate state estimates, which must be analyzed in
conjunction with the MPC controller \cite{MAR-DO:00}. A recent work
\cite{ADRDS-DE-TR-XP:22} uses an \emph{interval observer} to construct
a robust output-feedback MPC controller for linear systems. Interval
observers can be easily modified to generate predictions forward in
time, by removing the measurement term. Compared to existing methods
for tube-based output-feedback MPC, this method is more
computationally tractable. Our work aims to extend the interval
predictor approach to a class of nonlinear systems, so that it may be
used more readily in robot motion planning applications.

\textit{Statement of Contributions.} In this paper, we propose a novel
output-feedback MPC controller which provides guaranteed safety and
stability for a class of nonlinear systems. The MPC controller is
based on an interval observer and predictor, and is the first of its
kind for nonlinear systems. We show that, under some assumptions
related to the observability and controllability of the system, the
controller is able to maintain correct and bounded estimates and
predictions of the system state. Furthermore, by parameterizing the
control input with an initial stabilizing feedback term, we are able
to reduce the width of the predicted state intervals, thus reducing
the conservatism of the controller. The structure of this nonlinear
predictor is compatible with existing off-the-shelf MPC solvers,
meaning that the implementation (and computational complexity) is
similar to nominal ones. We prove that under some mild assumptions on
the cost function and terminal region, the controller achieves
recursive feasibility and asymptotic stability. Finally, we
demonstrate the controller on two example systems, a linear
continuously stirred tank reactor (CSTR) and a unicycle robot.

\emph{Notation:} The symbols $\real^n$, $\real^{n \times p}$,
$\mathbb{N}$, and $\mathbb{N}_n$ denote the $n$-dimensional Euclidean
space, the sets of $n$ by $p$ matrices, natural numbers (including 0),
and natural numbers from 1 to $n$, respectively. For
$M \in \real^{n \times p}$, $M_{ij}$ denotes $M$'s entry in the $i$'th
row and the $j$'th column,
$M^{\oplus}\triangleq \max(M,\mathbf{0}_{n,p})$,
$M^{\ominus}=M^{\oplus}-M$ and $|M|\triangleq M^{\oplus}+M^{\ominus}$,
where $\mathbf{0}_{n,p}$ is the zero matrix in $\real^{n \times
  p}$. For a matrix $A \in \real^{n \times n}$ we denote
\begin{gather*}
  \pmbox{A} \triangleq \begin{bmatrix}
    A^\oplus & -A^\ominus \\ -A^\ominus & A^\oplus
  \end{bmatrix} \in \real^{2n \times 2n}.
\end{gather*}
The symbol $\rho(M)$ denotes the spectral radius of $M$. For vectors
in $\real^n$, the comparisons $>$ and $<$ are considered
element-wise. Finally, an interval $[\ul{z},\ol{z}] \subset \real^n$
is the set of all real vectors $z \in \real^{n}$ that satisfies
$\ul{z} \le z \le \ol{z}$. Unless otherwise specified, for an interval
$[\ul{z}, \ol{z}]$, we define $\delta^z \triangleq \ol{z} - \ul{z}$
and for $[\ul{z}_k, \ol{z}_k]$,
$\delta^z_k \triangleq \ol{z}_k - \ul{z}_k$. A subscript $k$ always
denotes a time step, not a component of a vector.

\section{Preliminaries}
Next, we introduce some definitions and results from interval analysis
that will be leveraged in our interval observer and predictor
design. In the following definitions let $f$ be a function
$f : \X \subseteq \real^n \to \real^n$.

\begin{defn}
  The function $f$ is \emph{Jacobian sign-stable} (JSS) if the sign of
  each element of the Jacobian matrix $J_f(x)$ is constant for all
  $x \in \X$.
\end{defn}
\begin{prop}
  {\cite[Proposition 2]{MK-FS-SZY:22a}}
  \label{prop:JSS_decomp}
  If the Jacobian of $f$ is bounded, i.e., it satisfies
  $\ul{J}_f \le J_f(x) \le \ol{J}_f$ for all $x \in \X$, then $f$ can
  be written in additive-remainder form,
  \begin{align*}
     f(x) = A x + \mu(x),
  \end{align*}
  where the $(i,j)^{th}$ element of $A \in \real^{n \times n}$
satisfies
  \begin{align*}
    A_{ij} = (\ul{J}_f)_{ij} \ \text{ or } A_{ij} = (\ol{J}_f)_{ij}
  \end{align*}
  and the function $\mu$ is JSS.
  \bulletend
\end{prop}

\begin{defn}
  \cite[Definition 4]{LY-OM-NO:19}
  \label{def:decomp}
  A function $f_d: \X \times \X \to \real^{n}$ is a
  \emph{mixed-monotone decomposition function} for $f$ if
  \begin{enumerate}
  \item $f_d(x,x)=f(x)$,
  \item $f_d$ is monotonically increasing in its first argument,
  \item $f_d$ is monotonically decreasing in its second argument.
  \end{enumerate}
\end{defn}

\begin{prop}
  \cite[Proposition 4 \& Lemma 3]{MK-FS-SZY:22a}
  \label{prop:tight_decomp}
  Suppose $f$ is JSS and has bounded Jacobian. Then, the $i^{th}$
  component of a mixed-monotone decomposition function $f_d$ is given
  by
  \begin{align*}
    f_{d,i}(x_1,x_2) \triangleq f_i(D^i x_1 + (I_n - D^i) x_2),
  \end{align*}
  with
  $D^i =
  \mathrm{diag}(\max(\mathrm{sgn}((\ol{J}_f)_i),\mathbf{0}_{1,n}))$.
  Additionally, for any interval $[\ul{x}, \ol{x}] \subseteq \X$, it
  holds that $\delta^f_d \leq \ol{F}_f \delta^x$, where
  $\ol{F}_f \triangleq \ol{J}_f^\oplus + \ul{J}_f^\ominus$ and
  $\delta^f_d \triangleq f_d(\ul{x},\ol{x}) - f_d(\ol{x},
  \ul{x})$. \bulletend
\end{prop}

Consequently, applying Proposition \ref{prop:tight_decomp} to the JSS
remainder from Proposition \ref{prop:JSS_decomp} yields a tight
mixed-monotone decomposition function. Further details can be found in
\cite{MK-FS-SZY:22a}.  Finally, we review a well-known result which
bounds the product of a matrix with an interval.
\begin{prop}\cite[Lemma 1]{DE-TR-SC-AZ:13}\label{prop:bounding}
  Let $A \in \real^{p \times n}$ and
  $\ul{x} \leq x \leq \ol{x} \in \real^n$. Then,
  $A^\oplus \ul{x}-A^{\ominus}\ol{x} \leq Ax \leq
  A^\oplus\ol{x}-A^{\ominus}\ul{x}$.
\end{prop}

\section{Problem Formulation}
\subsection{System Dynamics}
Consider a discrete-time nonlinear system of the form
\begin{align}
  \label{eq:system}
  \begin{split}
    x_{k+1} &= f(x_k) + B u_k + w_k, \\
    y_k &= C x_k + v_k,
  \end{split}
\end{align}
where $x_k \in \real^n$ is the state, $u_k \in \real^m$ is the input,
$y_k \in \real^p$ is the output,
$w_k \in [\ul{w}, \ol{w}] \subset \real^n$ is the bounded process
noise and $v_k \in [\ul{v}, \ol{v}] \subset \real^p$ is the bounded
measurement noise. The function $f$, the matrices $B$ and $C$, and the
bounds on the noise are all known.

\subsection{Control Objective}
The objective is the minimization of a cost associated with the state
and control input of the system. The cost function is defined over a
horizon of length $N$. Let $\mathbf{x}$ and $\mathbf{u}$ denote the
$N$ and $(N-1)$-long trajectories of $x_k$ and $u_k$,
respectively. Then, for a stage cost $c$ and terminal cost $V_f$,
the cost is
\begin{gather}
  \label{eq:cost-proto}
  V_N(\mathbf{x}, \mathbf{u}) = \sum_{k = 0}^{N-1} c(x_k, u_k) + V_f(x_N).
\end{gather}
Since we do not have access to direct measurements of the state, this
cost function will be modified to use both \emph{estimates} and
\emph{predictions} of the state, as described in the next section.

In addition to minimizing the cost, we also require that the system
satisfy state and input constraints. In other words, for some (closed)
sets $\X \subseteq \real^n$ and $\U \subseteq \real^m$, the controller
should guarantee that $x_k \in \X$, and $u_k \in \U$ for all $k$.

\subsection{System Assumptions}
\begin{assumption}
  \label{ass:bounded}
  The function $f$ has a bounded Jacobian over $\mathcal{X}$, In other
  words, for all $x \in \mathcal{X}$,
  $\ul{J}_f \le J_f(x) \le \ol{J}_f$.
\end{assumption}

Using Proposition~\ref{prop:JSS_decomp}, we decompose $f$ into a
linear plus a JSS remainder term,
\begin{gather}
  \label{eq:decomp}
  f(x) = Ax + \mu(x).
\end{gather}

To obtain bounded predictions of the state trajectories, we require an
assumption related to stabilizability and detectability. As in
Proposition~\ref{prop:tight_decomp}, we define
$\ol{F}_\mu = \ol{J}_\mu^\oplus + \ul{J}_\mu^\ominus$.
\begin{assumption}
  \label{ass:L}
  There exist matrices $L$ and
  $K$ such that $\rho(|A-LC| + \ol{F}_\mu) < 1$ and
  $\rho(|A-BK| + \ol{F}_\mu) < 1$.
\end{assumption}
\begin{rem}
  Assumption~\ref{ass:L} requires that a \emph{linear} controller and
  observer be sufficient to stabilize the error system of the
  predictor. Intuitively, $\ol{F}_\mu$ captures the most extreme
  behavior of the nonlinearity, which must not destabilize the
  prediction error. We have studied observer gain design thoroughly in
  our previous work \cite[Theorem 2]{MK-SB-SM:24arxiv}, which can be
  adapted to compute the gains $L$ and $K$ and verify the assumption.
\end{rem}

\section{Interval Observer and Predictor Design}

\begin{figure*}[t!]
  \begin{gather}
    \label{eq:predictor}
    \begin{bmatrix}
      \ol{z}_{k, \ell+1} \\ \ul{z}_{k, \ell+1} \\
      \ol{e}_{k, \ell+1} \\ \ul{e}_{k, \ell+1}
    \end{bmatrix}
    = \begin{bmatrix}
      \pmbox{A - BK} & \pmbox{-BK} \\
      0 & \pmbox{A - LC}
    \end{bmatrix}
    \begin{bmatrix}
      \ol{z}_{k, \ell} \\ \ul{z}_{k, \ell} \\
      \ol{e}_{k, \ell} \\ \ul{e}_{k, \ell}
    \end{bmatrix}
    + \begin{bmatrix}
      \ol{w} + Bu'_{k,\ell} \\
      \ul{w} + Bu'_{k,\ell} \\
      \ol{v}_L - \hat{v}_L + \hat{w} - \ul{w} \\
      \ul{v}_L - \hat{v}_L + \hat{w} - \ol{w}
    \end{bmatrix}
    + \begin{bmatrix}
      \mu_d(\ol{z}_{k,\ell}, \ul{z}_{k,\ell}) \\
      \mu_d(\ul{z}_{k,\ell}, \ol{z}_{k,\ell}) \\
      \mu_d(\ol{\xi}_{k,\ell}, \ul{\xi}_{k,\ell})
      - \mu_d(\ul{z}_{k,\ell}, \ol{z}_{k,\ell}) \\
      \mu_d(\ul{\xi}_{k,\ell}, \ol{\xi}_{k,\ell})
      - \mu_d(\ol{z}_{k,\ell}, \ul{z}_{k,\ell})
    \end{bmatrix}
  \end{gather}
  \hrulefill
\end{figure*}

\subsection{Interval Observer}
In this section we introduce a basic interval observer, which
computes interval-valued state estimates using the measurement signal. The
interval observer has dynamics
\begin{align}
  \begin{split}
    \label{eq:observer}
    \begin{bmatrix}
      \ol{x}_{k+1} \\ \ul{x}_{k+1}
    \end{bmatrix}
    &= \pmbox{A - LC}
       \begin{bmatrix}
         \ol{x}_k \\ \ul{x}_k
       \end{bmatrix}
       + \begin{bmatrix}
         B \\ B
       \end{bmatrix} u_k
       + \begin{bmatrix}
         \ol{w} - \ul{v}_L \\
         \ul{w} - \ol{v}_L
       \end{bmatrix} \\
       &\quad + \begin{bmatrix}
         L \\ L
       \end{bmatrix} y_k
       + \begin{bmatrix}
         \mu_d(\ol{x}_k, \ul{x}_k) \\
         \mu_d(\ul{x}_k, \ol{x}_k)
      \end{bmatrix},
  \end{split}
\end{align}
where $\mu_d$ is a mixed-monotone decomposition function for $\mu$
(cf. Definition~\ref{def:decomp}),
$\ol{v}_L = L^\oplus\ol{v} - L^\ominus\ul{v}$, and
$\ul{v}_L = L^\oplus\ul{v} - L^\ominus\ol{v}$.

By construction, the estimates generated by this interval observer
will always contain the true state, and the estimate error will remain
bounded, as stated below.
\begin{prop}
  \label{prop:observer}
  If the initial condition of the interval
  observer~\eqref{eq:observer} satisfies
  $\ul{x}_0 \le x_0 \le \ol{x}_0$, then
  $\ul{x}_k \le x_k \le \ol{x}_k$ for all $k\ge 0$. Furthermore, if
  Assumption~\ref{ass:L} holds, then
  $\ol{x}_k - \ul{x}_k \le \Delta^x < \infty$ for some
  $\Delta^x \in \real$.
\end{prop}
\begin{proof}
  The statement $\ul{x}_k \le x_k \le \ol{x}_k$ holds by induction,
  using Proposition~\ref{prop:bounding} and the fact that $\mu_d$ is
  a mixed-monotone decomposition function for $\mu$.

  We show boundedness of the error term by defining
  $\delta^x_k = \ol{x}_k - \ul{x}_k$ and analyzing the comparison system
  \begin{gather*}
    \delta^x_{k+1} \le
    (|A - LC| + \ol{F}_\mu) \delta^x_k
    + \delta^w + \delta^{v_L}.
  \end{gather*}
  Since $\delta^x_k$ is guaranteed to be positive for all $k$, both
  the actual error dynamics and the comparison system are positive
  systems. Additionally, $|A-LC| + \ol{F}_\mu$ is stable by
  Assumption~\ref{ass:L} thus rendering the comparison system ISS with
  respect to $\delta^w + \delta^{v_L}$, which is bounded.
\end{proof}
\subsection{Interval Predictor}
This section outlines the design of the predictor, which computes
upper and lower bounds of the reachable set from a set of initial
conditions. At every time step, the initial condition of the predictor
is updated with the latest estimate from the interval observer. We
include an output feedback controller in the predictor system, which
reduces the widths of the predicted intervals.
\subsubsection{Feedback Parameterization}

We begin by parameterizing the control input as a feedback term plus a
feedforward term ($u'_k$), in order to reduce the width of the intervals
generated by the predictor system. We use an additional observer to
compute an estimate of the state, which is further leveraged to compute a
control input as follows:
\begin{align}
  \hat{x}_{k+1} &= (A-LC) \hat{x}_k + B u_k + \hat{w} - \hat{v}_L + Ly_k + \mu(\hat{x}_k),
                  \nonumber \\
  u_k &= -K \hat{x}_k + u'_k.
        \label{eq:fdbk}
\end{align}
Here, $L$ is the \emph{same} gain as in the interval observer,
$\hat{w} = \frac{1}{2}(\ol{w} + \ul{w})$, and
$\hat{v}_L = \frac{1}{2}(\ol{v}_L + \ul{v}_L)$.
\begin{rem}
  The structure of the observer in \eqref{eq:fdbk} is similar to that
  of the interval observer \eqref{eq:observer}. In fact, for linear
  systems (where $\mu$ is zero), $\hat{x}_k$ is the \emph{midpoint} of
  the interval observer's estimate, assuming it is initialized as
  such. For nonlinear systems, this is not the case, but
  $\hat{x}_k \in [\ul{x}_k, \ol{x}_k]$.  \bulletend
\end{rem}

With the controller \eqref{eq:fdbk} in place, the closed loop system
is
\begin{align}
  \label{eq:sys-cl}
  \begin{split}
    \begin{bmatrix}
      x_{k+1} \\ e_{k+1}
    \end{bmatrix}
    &= \begin{bmatrix}
      A - BK & -BK \\ 0 & A - LC
    \end{bmatrix}
       \begin{bmatrix}
         x_k \\ e_k
       \end{bmatrix} \\
     +& \begin{bmatrix}
       \mu(x_k) \\
       \mu(\hat{x}_k) - \mu(x_k) \\
     \end{bmatrix}
       + \begin{bmatrix}
         w_k + Bu'_k \\ Lv_k - \hat{v}_L + \hat{w} - w_k
       \end{bmatrix},
  \end{split}
\end{align}
where $e_k = \hat{x}_k - x_k$. Notice that although the feedback
controller \eqref{eq:fdbk} has the observer in the loop, the closed
loop dynamics does not contain the output $y_k$. This is advantageous,
since we will not have access to future values of the output when
designing the predictor. The fact that the output does not appear
means we can directly design a predictor for this closed loop system.

\subsubsection{Predictor Design}
Using a technique similar to the construction of the interval
observer, we can design a predictor for the system
\eqref{eq:sys-cl}. The only difference is the predictor does not have
access to any measurements, since it is predicting forward in time. To
describe the predictions, we introduce a new time index $\ell$, so
that for a variable $\xi$, $\xi_{k,\ell}$ is the prediction $\ell$
steps into the future, starting from time $k$. We will denote the
predicted upper and lower bounds of $x$ as $\ol{z}$ and $\ul{z}$, and
those of $e$ as $\ol{e}$ and $\ul{e}$. These predictions depend on the
(future) feedforward input $u'$, which will be determined by the MPC
controller. The predictor dynamics are given by \eqref{eq:predictor},
where $\ol{\xi}_{k,\ell} = \ol{z}_{k,\ell} + \ol{e}_{k,\ell}$ and
$\ul{\xi}_{k,\ell} = \ul{z}_{k,\ell} + \ul{e}_{k,\ell}$.

By construction, the predictor computes an overapproximation of the
reachable sets of $x_{k+\ell}$ and $e_{k+\ell}$, as long as it is
initialized correctly. We obtain the correct initialization from at
every $k$ from the interval observer. Now we are ready to state the
key result of this section, which states that system
\eqref{eq:predictor} generates correct predictions and that the
prediction error remains bounded for all time.

\begin{prop}
  \label{prop:predictor}
  Let $\ul{x}_k \le x_k \le \ol{x}_k$ for some $k$. If the initial
  condition of the predictor system \eqref{eq:predictor} satisfies
  \begin{gather*}
    \ul{z}_{k,0} \le \ul{x}_k, \
    \ol{z}_{k,0} \ge \ol{x}_k, \
    \ul{e}_{k,0} \le \ul{x}_k - \hat{x}_k, \
    \ol{e}_{k,0} \ge \ol{x}_k - \hat{x}_k,
  \end{gather*}
  then $\ul{z}_{k,\ell} \le x_{k+\ell} \le \ol{z}_{k,\ell}$ and
  $\ul{e}_{k,\ell} \le e_{k+\ell} \le \ol{e}_{k,\ell}$ for all
  $\ell \ge 0$. Furthermore, if Assumption~\ref{ass:L} (on the
  existence of stabilizing gains) holds, the widths of the prediction
  intervals are uniformly bounded for all $\ell \ge 0$, i.e.,
  \begin{gather*}
    \delta^z_{k,\ell} = \ol{z}_{k, \ell} - \ul{z}_{k, \ell} \le \Delta^z, \quad
    \delta^e_{k,\ell} = \ol{e}_{k, \ell} - \ul{e}_{k, \ell} \le \Delta^e,
  \end{gather*}
  where $\Delta^z < \infty$ and $\Delta^e < \infty$ are scalars.
\end{prop}

\begin{proof}
  Similar to Proposition~\ref{prop:observer}, the statement that
  $\ul{z}_{k,\ell} \le x_{k,\ell} \le \ol{z}_{k,\ell}$ and
  $\ul{e}_{k,\ell} \le e_{k,\ell} \le \ol{e}_{k,\ell}$ holds by
  induction, using Proposition~\ref{prop:bounding} and the fact that
  $\mu_d$ is a mixed-monotone decomposition function for $\mu$. After some
  manipulation, we can write a comparison system for the dynamics of
  the width of the prediction sets, defining $\tilde{A} \triangleq \begin{bmatrix}
      |A - BK| + \ol{F}_\mu & |BK| \\ 0 & |A - LC| + \ol{F}_\mu
    \end{bmatrix}$, as follows:
  \begin{align*}
    \begin{bmatrix}
      \delta^z_{k,\ell+1} \\ \delta^e_{k,\ell+1}
    \end{bmatrix}
    \le \begin{bmatrix}
         \delta^w \\  |L|\delta^v + \delta^w
       \end{bmatrix} + \tilde{A}
     \begin{bmatrix} \delta^z_{k,\ell} \\ \delta^e_{k,\ell}
    \end{bmatrix},
  \end{align*}
  which is a positive system by construction. By
  Assumption~\ref{ass:L},
  $\rho(\tilde{A}) < 1$, which, since
  $\delta^w$ and $\delta^v$ are constant, ensures that the prediction
  error remains bounded.
\end{proof}

\section{Output-feedback MPC}

In this section we combine the interval observer and predictor to
design an output-feedback model predictive controller. We begin by
describing the cost function which the controller will attempt to
minimize.

\subsection{Cost function}
Now that the predictor and its associated notation have been defined,
we can correct the prototype cost function \eqref{eq:cost-proto} into
a form which can be evaluated using only the available predictions. The
cost function is then
\begin{gather*}
  V_N(\mathbf{z}, \mathbf{u}') = \sum_{\ell = 0}^{N-1} c(\hat{z}_{k,\ell}, u'_{k,\ell})
  + V_f(\hat{z}_{k,N}),
\end{gather*}
where $\mathbf{z}$ and $\mathbf{u}'$ are the $N$ and $(N-1)$-long
trajectories of $(\ul{z}_{k,\ell}, \ol{z}_{k,\ell})$ and
$u'_{k,\ell}$, respectively. The term
$\hat{z}_{k,\ell} = \frac{1}{2}(\ol{z}_{k,\ell} + \ul{z}_{k,\ell})$ is
the midpoint of the predicted interval at step $\ell$.
\subsection{Terminal region and controller}

To ensure recursive feasibility of the MPC scheme, we constrain the
system to reach a target set $\X_f$ within a horizon of length
$N$. Furthermore, we assume that there exists a feedback controller
which renders $\X_f$ forward invariant.\footnote{As is common in the
  MPC literature, the existence of this controller is used to prove
  recursive feasibility. In practice, the system does not need to
  \emph{actually} switch to this controller once the target set is
  reached. In our numerical examples we leave the MPC controller
  running in the target set.}
\begin{assumption}
  \label{ass:terminal}
  There exists a terminal set $\X_f$ and an associated controller
  $K_f : \X_f \to \U$ such that for all $x \in \X_f$,
  $f(x, K_f(x), \hat{w}) \in \X_f$.  Furthermore, the terminal set
  satisfies $\X_f \oplus \mathcal{B}_{\Delta^z} \subset \X$, where
  $\mathcal{B}_{\Delta^z}$ denotes the infinity-norm ball of radius
  $\Delta^z$ and $\oplus$ denotes the Minkowski sum.  \bulletend
\end{assumption}

We will use this terminal set to constrain the midpoint of the final
step of the predicted trajectory. Therefore we also require that the
entire predicted interval is safe whenever the midpoint is inside the
target set, which is ensured by the second condition of the
assumption.

\subsection{Optimization problem}
Our goal is to compute a sequence of feedforward inputs which
minimizes the cost while ensuring satisfaction of the state and
control constraints for \emph{any} realization of the state and
measurement disturbances. The optimization problem is
\begin{align}
  \label{eq:mpc}
  \begin{split}
    \min_{\mathbf{z}, \mathbf{e}, \mathbf{u'}}& \ V_N(\mathbf{z}, \mathbf{u'}) \\
    \text{s.t.}& \ \text{\eqref{eq:predictor}}, \ \hat{z}_{k,N}  \in \X_f,
      \  \ul{z}_{k,0} \le \ul{x}_k, \
      \ol{z}_{k,0} \ge \ol{x}_k, \\
    & \ul{e}_{k,0} \le \ul{x}_k - \hat{x}_k, \
      \ol{e}_{k,0} \ge \ol{x}_k - \hat{x}_k, \\
    & [\ul{z}_{k,\ell}, \ol{z}_{k,\ell}] \subset \X, \
      [\ul{u}_{k,\ell}, \ol{u}_{k,\ell}]  \subset \U, \\
    & \begin{bmatrix}
        \ul{u}_{k,\ell} \\ \ol{u}_{k,\ell}
      \end{bmatrix} =
      \pmbox{K}
      \begin{bmatrix}
        \ul{z}_{k,\ell} \\ \ol{z}_{k,\ell}
      \end{bmatrix} +
      \begin{bmatrix}
        u'_{k,\ell} \\ u'_{k,\ell}
      \end{bmatrix},
       \forall \ell \in \{0, \dots, N\}.
  \end{split}
\end{align}

Algorithm~\ref{alg:mpc} describes the proposed improved robust output feedback MPC (IROF-MPC) scheme,
which involves recursively updating the interval observer using
measurements, and using the interval estimate to initialize the
optimization.

\begin{algorithm}[t]
  \caption{Improved Robust Output Feedback-MPC (IROF-MPC) Scheme at Time Step $k$.}
  \label{alg:mpc}
  \begin{algorithmic}[1]
    \renewcommand{\algorithmicrequire}{\textbf{Input:}}
    \renewcommand{\algorithmicensure}{\textbf{Output:}}
    \Require $y_k$, $\ul{z}_{k-1,1}$, $\ol{z}_{k-1,1}$;
    \textbf{Output:} $u_k$, $\ul{z}_{k,1}$, $\ol{z}_{k,1}$
    \State Update $\ul{x}_k$ and $\ol{x}_k$ using $y_k$ and \eqref{eq:observer}.
    \State $\ul{x}_k \gets \max\{\ul{x}_k, \ul{z}_{k-1,1}\}$,
    $\ol{x}_k \gets \min\{\ol{x}_k, \ol{z}_{k-1,1}\}$
    \State Solve \eqref{eq:mpc} to determine $u'_k$
    \State $u_k \gets K \hat{x}_k + u'_k$ \\
    \Return $u_k$
  \end{algorithmic}
\end{algorithm}
\subsection{Feasibility and stability}
\begin{thm}
  Let Assumptions~\ref{ass:bounded}-\ref{ass:terminal} hold. Then
  the MPC controller given by Algorithm~\ref{alg:mpc} is recursively
  feasible, meaning that if \eqref{eq:mpc} is feasible at $k=0$, then it is
  feasible for all $k > 0$. Furthermore, the state and control
  satisfy $x_k \in \X$ and $u_k \in \U$, respectively, for all $k \ge 0$.
\end{thm}
\begin{proof}
  Assume there is a sequence of feedforward inputs
  $\{u'_{0,\ell}\}_{\ell=0}^{N-1}$ solving the optimization problem
  \eqref{eq:mpc} at $k = 0$. At $k = 1$, a new sequence
  $\{u'_{1,\ell}\}_{\ell=0}^{N-1}$ can be constructed by letting
  $u'_{1,\ell} = u'_{0, \ell+1}$ for $\ell \in \{0, \dots, N-2\}$. The
  final term in the sequence is given by the feedback law from the
  target set, $u'_{1, N-1} = K_f(\hat{z}_{0,N-1})$, which is
  guaranteed to keep $\hat{z}_{1,N}$ in the target set by
  Assumption~\ref{ass:terminal}. This proves that the target set
  constraint is satisfied. The input constraints are also satisfied by
  this construction.

  To prove that the state constraints are satisfied, note that since
  $[\ul{z}_{1,0}, \ol{z}_{1,0}] \subset [\ul{z}_{0,1}, \ol{z}_{0,1}]$,
  it must be true that
  $[\ul{z}_{1,\ell}, \ol{z}_{1,\ell}] \subset [\ul{z}_{0,\ell+1},
  \ol{z}_{0,\ell+1}] \subset \X$ for $\ell \in \{0, \dots, N-1\}$.
  Additionally, by Assumption~\ref{ass:terminal}, since
  $\hat{z}_{1,N} \in \X_f$, $[\ul{z}_{1,N}, \ol{z}_{1,N}] \subset \X$.
  This concludes the proof that $\{u'_{1,\ell}\}_{\ell=0}^{N-1}$ is a
  feasible solution of the optimization problem. This procedure can be
  repeated recursively.

  The satisfaction of the constraints $x_k \in \X$ and $u_k \in \U$
  follows directly from the constraints of \eqref{eq:mpc}, and the
  fact that by Proposition~\ref{prop:predictor}, the predicted state
  intervals are guaranteed to contain the future states of the system.
\end{proof}

While this theorem guarantees that the state and input always satisfy
their respective constraints, it does not guarantee asymptotic
stability to the target set. Doing so requires some additional
assumptions on the cost function and target set. The cost function
should be positive definite about an equilibrium point inside the
target set, so that it may act as a candidate Lyapunov function for
the closed-loop system.
\begin{assumption}
  \label{ass:equilibrium}
  The target set $\X_f$ contains the origin, which is an equilibrium
  point of the dynamics \eqref{eq:system}. Furthermore, the stage cost
  $c(\cdot, u)$ is positive definite for all $u \in \U$, and the
  terminal cost $V_f$ is positive definite.
\end{assumption}

The following theorem gives a sufficient condition for the stability
of the closed-loop in a certain region of attraction. This region of
attraction can be made larger by extending the prediction horizon of
the controller.
\begin{thm}
  Let Assumptions~\ref{ass:bounded}-\ref{ass:equilibrium} hold. If, for
  all $x \in \X_f$,
  \begin{gather*}
    V_f(f(x, K_f(x), \hat{w})) - V_f(x) + c(x, K_f(x)) \le 0,
  \end{gather*}
  then the closed-loop dynamics of $\hat{z}_{k,0}$ with the MPC
  controller described by Algorithm~\ref{alg:mpc} is asymptotically
  stable. The region of attraction is the set of initial conditions
  for which the optimization problem~\eqref{eq:mpc} is feasible.
\end{thm}
\begin{proof}
  The proof is similar to classical results from the MPC literature,
  where the cost function is used to create a Lyapunov function for
  the closed loop system. Due to space concerns, we refer the reader
  to \cite{JR-DM-MD:17} for more details.
\end{proof}

This result relies on the assumption that the terminal controller is
able to attain a cost decrease within the terminal set. This
assumption is standard in the literature on nonlinear MPC, and the
computation of such a controller is beyond the scope of this paper. We
refer the reader to \cite{HC-FA:98} and \cite{WC-JO-DJB:03} for
discussions of this problem in the context of continuous-time systems;
similar results hold in discrete-time.

\section{Numerical Examples}
\subsection{Linear CSTR}
\label{sec:comp}
In this section we compare our improved robust output feedback MPC (IROF-MPC) scheme with a recent work on interval
MPC for linear systems \cite{ADRDS-DE-TR-XP:22}~(DERP). We use the linear model (i.e., $f(x)=Ax$)
of a continuous stirred tank reactor (CSTR) from \cite{ADRDS-DE-TR-XP:22}, where the system matrices are:
\begin{gather*}
  A =
  \begin{bmatrix}
    0.745 & -0.002 \\ 5.610 & 0.780
  \end{bmatrix}, \
  B =
  \begin{bmatrix}
    5.6 \times 10^{-6} \\ 0.464
  \end{bmatrix}, \
  C =
  \begin{bmatrix}
    0 & 1
  \end{bmatrix}.
\end{gather*}
Since the system is linear, the remainder function $\mu$ is the zero
function. The constraints are $\X = [-0.4, 0.4] \times [-25, 25]$ and
$\U = [-15, 15]$, and the disturbance bounds are
$[\ul{w}, \ol{w}] = \alpha([-0.02, 0.02] \times [-0.4, 0.4])$ and
$[\ul{v}, \ol{v}] = \beta[-0.1, 0.1]$, where $\alpha$ and $\beta$ will
be varied to test different magnitudes of noise. The observer gain is
$L = \begin{bmatrix} -0.002 & 0.390 \end{bmatrix}^\top$ and the interval
observer is initialized with $\ul{x}_0 = [-0.1, -0.05]$ and
$\ol{x}_0 = [0.1, 0.05]$.

The controller is designed to track a setpoint outside the feasible
set, $x_r = [-0.25, 27.3]$. The cost function is quadratic in the
error $e_r = x - x_r$, $c(x, u) = e_r^\top H e_r + Ru^2$, with
$H = 100I$ and $R = 0.01$. The horizon is $N = 10$.

We conducted several trials with modified noise bounds, varying
$\alpha$ and $\beta$. The noise is uniformly distributed over the
bounding interval. In every trial, we run our method and the one from
\cite{ADRDS-DE-TR-XP:22} with identical realizations of the noise, in
order to compare their performance. Figure~\ref{fig:mpc-comp} shows
the result of one of those trials. Figure~\ref{fig:alpha-beta} shows
the results (the mean squared tracking error on a trajectory of 50
time steps) of varying one of either $\alpha$ and $\beta$ while
keeping the other parameter fixed. The improved performance of our
method becomes clearer as the noise magnitude is increased. This
increase in performance (i.e., decrease in tracking error) is due to
our inclusion of a closed-loop control inside our predictor, which
limits future uncertainty, narrowing the width of the prediction. This
allows the system trajectory to get closer to the inequality
constraint while still ensuring its satisfaction.
\begin{figure}[]
  \centering
  \includegraphics[width=\columnwidth]{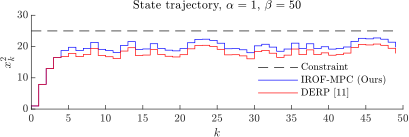}
  \caption{Trajectories of $x^2_k$ (the second state component) for
    one trial of the experiment in Section~\ref{sec:comp}, using the
    approach from \cite{ADRDS-DE-TR-XP:22} (DERP) and IROF-MPC. The
    objective is to operate as close to the constraint as possible.}
  \label{fig:mpc-comp}
\end{figure}
\begin{figure}[]
  \centering
  \includegraphics[width=\columnwidth]{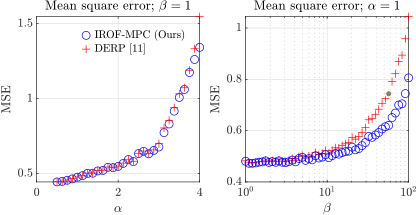}
  \caption{Mean square error of several experiments varying the
    magnitude of state and process noise, using the method from
    \cite{ADRDS-DE-TR-XP:22} (DERP) and our IROF-MPC.}
  \label{fig:alpha-beta}
\end{figure}

\subsection{Robotic Exploration}
In this section we demonstrate the effectiveness of our controller on
a scenario involving a unicycle robot exploring and measuring an
uncertain environment. For this purpose, we use our MPC controller
along with the recently developed Behavioral
Entropy~\cite{AS-CN-SM:24} in the cost function to incentivize the
agent to explore while still reaching the target state.

After performing the JSS decomposition, the time-discretized
($\Delta_t = 0.1$) unicycle robot has dynamics
\begin{gather*}
  A =
  \begin{bmatrix}
    1-\delta & 0 & \Delta_t & \Delta_t  \\
    0 & 1-\delta & \Delta_t & \Delta_t \\
    0 & 0 & 1 & 0 \\
    0 & 0 & 0 & 1-b_v
  \end{bmatrix}, \
  B =
  \begin{bmatrix}
    0 & 0 \\
    0 & 0 \\
    \Delta_t & 0 \\
    0 & \Delta_t
  \end{bmatrix}, \\
  \mu_1(x) = \Delta_t
   (x^4 \cos(x^3) - x^3 - x^4), \ \mu_3(x) = \mu_4(x) = 0, \\
   \mu_2(x) = \Delta_t(x^4 \sin(x^3) - x^3 - x^4).
\end{gather*}

We represent the environment as a compact set
$\E \subset \real^2$. We discretize $\E$ using a square grid
and call the resulting set $\mathcal{D}$. Associated to the grid
$\mathcal{D}$ is an occupancy function
$f_{occ} : \mathcal{D} \to [0,1]$, which represents the probability
that a cell contains the quantity of interest. The occupancy map
$\mathcal{M}$ is the discrete field of $f_{occ}$ over
$\mathcal{D}$. The quantity represented by the occupancy map depends
on the goal of the robot, but for example could be the presence of
survivors in a search and rescue mission or the presence of a toxic
gas. All cells in the initial occupancy map $\mathcal{M}_0$ have a
value of $0.5$.  The goal of the robot is to reduce the uncertain area
by measuring grid cells and updating the occupancy map
$\mathcal{M}_k$. We assume that the robot ``measures'' a grid cell
whenever its location is within that cell.

The total measure of uncertainty throughout the area is given by an
entropy function which is evaluated over the occupancy map. More
specifically, for any admissible generalized entropy $H$ and occupancy
map $\mathcal{M}$, the total uncertainty is
$\bar{H}(\mathcal{M}) = \sum_{x\in\mathcal{D}} H(p(x))$, where $p(x)$
denotes the Bernoulli distribution associated with the occupancy at
$x$. In order for our controller to converge to a target region, we
require that the uncertainty is always reduced over time.  In other
words, the occupancy map should satisfy
$\bar{H}(\mathcal{M}_{k+1}) \le \bar{H}(\mathcal{M}_k)$ for all
$k \ge 0$.

Having defined the generalized entropy, the cost is
\begin{gather*}
  V_N(\mathbf{z}, \mathbf{u}') \hspace{-.1cm}= \hspace{-.15cm}{\sum}_{\ell = 0}^{N-1} \left[c(\hat{z}_{k,\ell}, \hat{u}_{k,\ell})
   \hspace{-.05cm} - \hspace{-.05cm}\lambda U(\hat{z}_{k,\ell}; \mathcal{M}_k)\right] \hspace{-.1cm}+\hspace{-.1cm} V_f(\hat{z}_{k,N}),
\end{gather*}
where the function $U$ represents the information gained by measuring
the grid cell in $\mathcal{D}$ which contains $\hat{z}_k$, given the
current occupancy map $\mathcal{M}_k$, $U(x; \mathcal{M}) = H(p(x))$,
where, as before, $p(x)$ is the Bernoulli distribution associated with
the occupancy at $x$. The parameter $\lambda$ can be tuned to control
the emphasis on exploration vs. progress toward the goal.

For comparison, we also implement a nominal MPC controller, which uses
the observer-based controller \eqref{eq:fdbk} in conjunction with a
simplified optimization problem, with the same objective, state and
input constraints, and predictions given by
$\hat{z}_{k+1} = f(\hat{z}_k) + B (u'_k - K \hat{z}_k) + \hat{w}, \
\hat{z}_{k,0} = \hat{x}_k$.

Figure~\ref{fig:uni-x} shows the closed-loop trajectories of the
simulation. Evidently the nominal controller causes the system to
violate the (non-convex) obstacle constraint, because it does not
account for the noise on the state and output. On the other hand, the
proposed method avoids the obstacle due to the interval predictions
which are robust to noise.

Table~\ref{tbl:comp} lists the computation times (using Acados
\cite{RV:21} on an i5-1240P @ 4.4GHz with 64GB RAM) averaged over 100
timesteps, with horizons of $N=35$ and $N=100$. Each trial is
warm-started with a solution from a generic NLP solver. In all cases,
the computational cost of our method is similar to that of the nominal
controller.
\begin{figure}[]
  \includegraphics[width=\columnwidth]{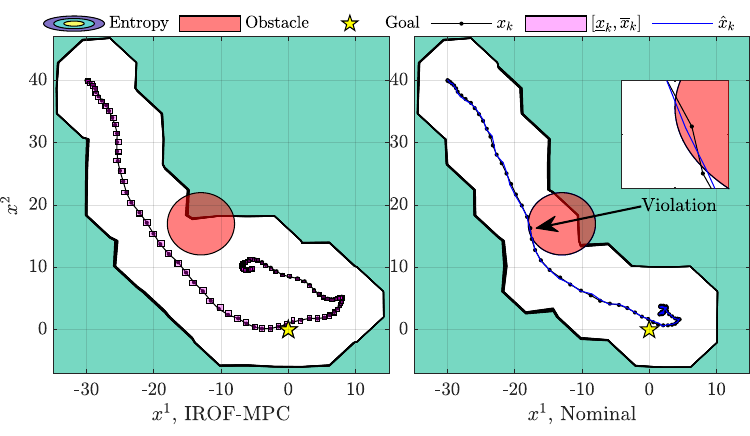}
  \caption{Comparison of nominal output-feedback MPC versus IROF-MPC
    over $100$ time steps in the robot exploration scenario.}
  \label{fig:uni-x}
\end{figure}
\begin{table}[H]
  \caption{Computation times over 100 timesteps, in milliseconds}
  \label{tbl:comp}
  \begin{tabular}{ccccccc}
    \multicolumn{1}{c|}{}         & \multicolumn{3}{c|}{$N = 35$}                                           & \multicolumn{3}{c}{$N = 100$}                                      \\ \cline{2-7}
    \multicolumn{1}{c|}{}         & max                  & avg                  & \multicolumn{1}{c|}{min}  & max                  & avg                  & min                  \\ \hline\hline
    \multicolumn{1}{r|}{IROF-MPC} & 3.89                 & 1.16                 & \multicolumn{1}{c|}{0.58} & 12.743               & 1.96                 & 1.19                 \\
    \multicolumn{1}{r|}{Nominal}  & 4.13                 & 1.36                 & \multicolumn{1}{c|}{1.02} & 11.82                & 2.72                 & 2.018                \\
    \multicolumn{1}{l}{}          & \multicolumn{1}{l}{} & \multicolumn{1}{l}{} & \multicolumn{1}{l}{}      & \multicolumn{1}{l}{} & \multicolumn{1}{l}{} & \multicolumn{1}{l}{}
  \end{tabular}
\end{table}
\section{Conclusion and Future Work}
This paper introduced a robust output-feedback MPC controller. We
described an interval observer and predictor, which were then proven
to generate correct and stable estimates and predictions. In addition,
we showed that the closed loop system under the MPC controller is
recursively feasible, safe, and, with some additional assumptions on
the cost function, asymptotically stable. We showed that our predictor
design, which includes a stabilizing feedback term, was able to
outperform existing methods, especially in the presence of noise. We
further demonstrated its effectiveness in robotic exploration, where
the optimization could be solved over a long horizon.

In the future we plan to investigate the robotic exploration scenario
more thoroughly, including different cost functions. We also plan to
design optimal observer and controller gains which result in the best
closed-loop performance.

\bibliographystyle{IEEEtran}
\bibliography{alias,SM,SMD-add}

\end{document}